\documentclass[twocolumn,conference]{IEEEtran}

\usepackage[normalem]{ulem}
\usepackage{multirow}
\usepackage{textcomp}
\usepackage{amsfonts}

\DeclareFontFamily{U}{wncy}{}
\DeclareFontShape{U}{wncy}{m}{n}{<->wncyr10}{}
\DeclareSymbolFont{mcy}{U}{wncy}{m}{n}
\DeclareMathSymbol{\Sh}{\mathord}{mcy}{"58}

\usepackage[all,cmtip]{xy}
\usepackage{amssymb, amsmath, amscd, amsthm, amsfonts}
\newtheorem{thm}{Theorem}

\DeclareMathOperator{\tr}{tr}

\usepackage{cite}

\ifCLASSINFOpdf
   \usepackage[pdftex]{graphicx}
\else
\fi

\usepackage{tikz-cd}
\usepackage{mathtools}

\usepackage{algorithm}
\usepackage[noend]{algpseudocode}
\algnewcommand{\LineComment}[1]{\State \(\triangleright\) #1}

\ifCLASSOPTIONcompsoc
  \usepackage[caption=false,font=normalsize,labelfont=sf,textfont=sf]{subfig}
\else
  \usepackage[caption=false,font=footnotesize]{subfig}
\fi
\usepackage{url}

\begin{document}
%
\title{On-the-fly Adaptive $k$-Space Sampling for Linear MRI Reconstruction Using Moment-Based Spectral Analysis}


\author{Evan~Levine,~\IEEEmembership{Student Member,~IEEE,}
        ~Brian~Hargreaves
\thanks{E. Levine is with the Departments
of Electrical Engineering and Radiology, Stanford University, Stanford, CA, 94305 USA e-mail: (egl@stanford.edu)}
\thanks{B. Hargreaves is with the Department of Radiology, Stanford University, Stanford, CA, 94305 USA e-mail: (bah@stanford.edu)} 
\thanks{Manuscript received xxx; revised xxx.}}


%

\maketitle

\begin{abstract} 
    In high-dimensional magnetic resonance imaging applications, time-consuming, sequential acquisition of data samples in the spatial frequency domain ($k$-space) can often be accelerated by accounting for dependencies along imaging dimensions other than space in linear reconstruction, at the cost of noise amplification that depends on the sampling pattern. Examples are support-constrained, parallel, and dynamic MRI, and $k$-space sampling strategies are primarily driven by image-domain metrics that are expensive to compute for arbitrary sampling patterns. It remains challenging to provide systematic and computationally efficient automatic designs of arbitrary multidimensional Cartesian sampling patterns that mitigate noise amplification, given the subspace to which the object is confined. To address this problem, this work introduces a theoretical framework that describes local geometric properties of the sampling pattern and relates these properties to a measure of the spread in the eigenvalues of the information matrix described by its first two spectral moments. This new criterion is then used for very efficient optimization of complex multidimensional sampling patterns that does not require reconstructing images or explicitly mapping noise amplification. Experiments with \textit{in vivo} data show strong agreement between this criterion and traditional, comprehensive image-domain- and $k$-space-based metrics, indicating the potential of the approach for computationally efficient (on-the-fly), automatic, and adaptive design of sampling patterns.

\end{abstract}


%
\IEEEpeerreviewmaketitle


\section{Introduction}
\IEEEPARstart{M}{agnetic} resonance imaging (MRI) scanners acquire data samples in the spatial frequency domain ($k$-space) of an object. In Cartesian MRI, the Nyquist sampling theorem dictates that the sample spacing and extent should correspond to the field-of-view and resolution of the acquisition, necessitating time-consuming sequential scanning of $k$-space. In higher-dimensional MRI applications, data is acquired along additional dimensions beyond space, such as time, echoes, or receive channels. To enable acceleration relative to the nominal rate, a signal model is often incorporated in the reconstruction to account for dependencies along dimensions other than space. In many cases, the signal model is defined per voxel and consists of a predefined subspace model. Assuming Gaussian noise, an unbiased estimate is then a pseudoinverse, for which errors can be characterized by a nonuniform noise level. Noise amplification due to this linear inversion depends only on the $k$-space sample distribution and reconstruction subspace but not the underlying object.

Many well-known linear reconstructions confine the signal in each voxel to a predefined subspace. A simple example is support-constrained imaging, which confines voxels outside of the support to the zero-dimensional subspace. The problem has been studied in the literature on nonuniform sampling and reconstruction of multiband signals \cite{jerri1977shannon, zayed1993advances, john1996sampling, venkataramani2001optimal, marks2012advanced}, with general theory developed by Landau \cite{landau1967necessary}. Some specialized approaches in MRI consider support regions for which an optimal sampling pattern can be selected \cite{Madore,Tsao2003}. A second example is parallel MRI, which uses multiple RF receive channels and relies on dependencies between channels to accelerate. In Sensitivity Encoding (SENSE) \cite{pruessmann1999sense}, the profile in each voxel across channels is restricted to a one-dimensional subspace spanned by the channel sensitivities and is scaled by the underlying magnetization. Further consideration has to be made when sampling can be varied along additional dimensions, which we index here by $t$. One example is dynamic MRI, where signal-intensity-time curves in each voxel are represented by a low-dimensional subspace \cite{liang2007spatiotemporal}. Other examples arise in parametric mapping \cite{petzschner2011fast,huang2013t2,zhao2015accelerated,tamir2016t2}, and Dixon fat and water MRI \cite{dixon1984simple}, to name a few.

Many conventional MRI acceleration methods use uniform sampling patterns, but the same inversions can be applied to nonuniformly sampled $k$-$t$-space by using iterative reconstruction. Nonuniform sampling allows acceleration factors to be fine-tuned and can be desirable for non-Cartesian MRI, contrast enhancement, motion, or regularized reconstruction.

Although a complete characterization of reconstruction errors is possible with the use of an image-domain geometry (g)-factor metric, g-factor is computationally expensive to estimate for arbitrary sampling patterns using pseudo-multiple-replica-type methods \cite{robson2008comprehensive,riffe2007snr} and does not directly explain the sources of noise amplification in terms of correlations in $k$-$t$-space. To address the latter problem, the inverse problem can be posed as approximation in a reproducing kernel Hilbert space, which provides a characterization in $k$-space \cite{athalye2015parallel}, but it still demands a computationally expensive procedure to derive error bounds. To our knowledge, a formalism that provides similar insights and is amenable to the computationally efficient generation of arbitrary $k$-$t$-space sampling patterns is still unavailable.

One nonuniform sampling strategy that is commonly used is Poisson-disc sampling, which has the property that no two samples are closer than a specified distance, despite the randomization that is motivated by compressed sensing \cite{lustig2010spirit}. Poisson-disc sampling has empirically performed well in parallel imaging, even without regularization, and it has been used in hundreds of publications to date, \cite{vasanawala2011practical,gdaniec2014robust,lebel2014highly,hsiao2012rapid}, to name a few. Poisson-disc sampling has been motivated by the use of a synthesizing kernel in $k$-space and the ``fit" between a point spread-function (PSF) and coil array geometry \cite{vasanawala2011practical}. However, criteria for the coil-PSF matching have not been formally described.

Some optimality criteria have been utilized for sampling and attempt to minimize a \textit{spectral moment}, either the trace of the inverse \cite{xu2005optimal} or the maximum reciprocal eigenvalue \cite{curtis2015optimal} of the information matrix. Both have prohibitive computational complexity for large-sized problems (e.g. 3D cartesian MRI at moderate resolution) because of how they probe a correspondingly large unstructured inverse matrix. Related methods approximate the inverse matrix and exploit the local nature of reconstruction in $k$-space \cite{samsonov2008optimality} and have been applied to the design of radial sampling patterns \cite{samsonov2009automatic}.

In this work, we use a moment-based spectral analysis to describe a new criterion for linear reconstruction that avoids inverting the information matrix and instead minimizes the variance in the distribution of its eigenvalues. Despite the use of eigenvalues in the criterion, this approach still provides simplified and interpretable theoretical insights into noise amplification in $k$-$t$-space. Local correlations in $k$-space are summarized by a weighting function derived only from system parameters. The new optimality criterion then specifies a corresponding cross-correlation or \textit{differential distribution} of sampling patterns.

The new analytical relationships described in this paper are developed into fast algorithms for generating adaptive sampling that may be arbitrary or structured. Since these algorithms do not explicitly map noise amplification or reconstruct images, they have extremely short run times, often sufficient for real-time interactive applications. We make our implementation available to the community.\footnote{\url{https://github.com/evanlev/dda-sampling}} Numerical experiments are performed to evaluate noise amplification for the derived sampling patterns. This is critical for constrained reconstruction, where conditioning is a major consideration in selecting sampling patterns.

\section{Theory}

\subsection{Differential sampling domain}
For notational simplicity, consider sampling on a finite one-dimensional Cartesian $k$-space grid of size $N$. Let $\{k_n^{(t)} \}_{n=0,t=0}^{N_t-1, T-1}$ be sample locations from one of $T$ sampling patterns. In the subsequent notation, $t$ and $t^\prime$ arguments are omitted in cases where $T=1$. Define the $t$\textsuperscript{th} sampling pattern as the function
\begin{equation}
s^{(t)}(k) = \sum_{n=0}^{N_t-1} \delta(k - k^{(t)}_n),
\end{equation}
where $\delta$ is the discrete Kronecker delta function, extended with period $N$. Define the point-spread function (PSF) for the sampling pattern as
\begin{equation}
PSF_t = \mathcal{F}^{-1} \{ s^{(t)} \},
\label{eq:psf_ft1}
\end{equation}
where $\mathcal{F}$ is a discrete Fourier transform. 

Next, define a distribution of inter-sample differences, or \textit{differential distribution}, of the sample locations, which is a cross-correlation of the sampling patterns:
\begin{align}
p(\Delta k, t, t^\prime) &= \sum_{n=0}^{N_t-1} \sum_{m=0}^{N_{t^\prime}-1} \delta(\Delta k - k_n^{(t)} + k_m^{(t^\prime)}) \\
                         &= (s^{(t)} \star s^{(t^\prime)})(\Delta k)
\label{eq:p_autocorr}
\end{align}

A similar continuous-domain differential distribution was previously introduced in computer graphics \cite{wei2011differential, zhou2012point,heck2013blue} to generate a single sampling pattern with a user-defined power spectrum.

From (\ref{eq:p_autocorr}) and (\ref{eq:psf_ft1}), the correlation property of the Fourier transform relates the differential distribution to the product of PSF's, the squared magnitude of the PSF in the single-time-point case:
\begin{equation}
p(\Delta k, t, t^\prime) = N \mathcal{F}\{ PSF_t(r) PSF_{t^\prime}^*(r) \}(\Delta k)
\label{eq:psf_ft2}
\end{equation}

Fig. \ref{fig:relationships} shows this relationship in an example. 
\begin{figure}[!t]
\centering
\includegraphics[width=0.49\textwidth]{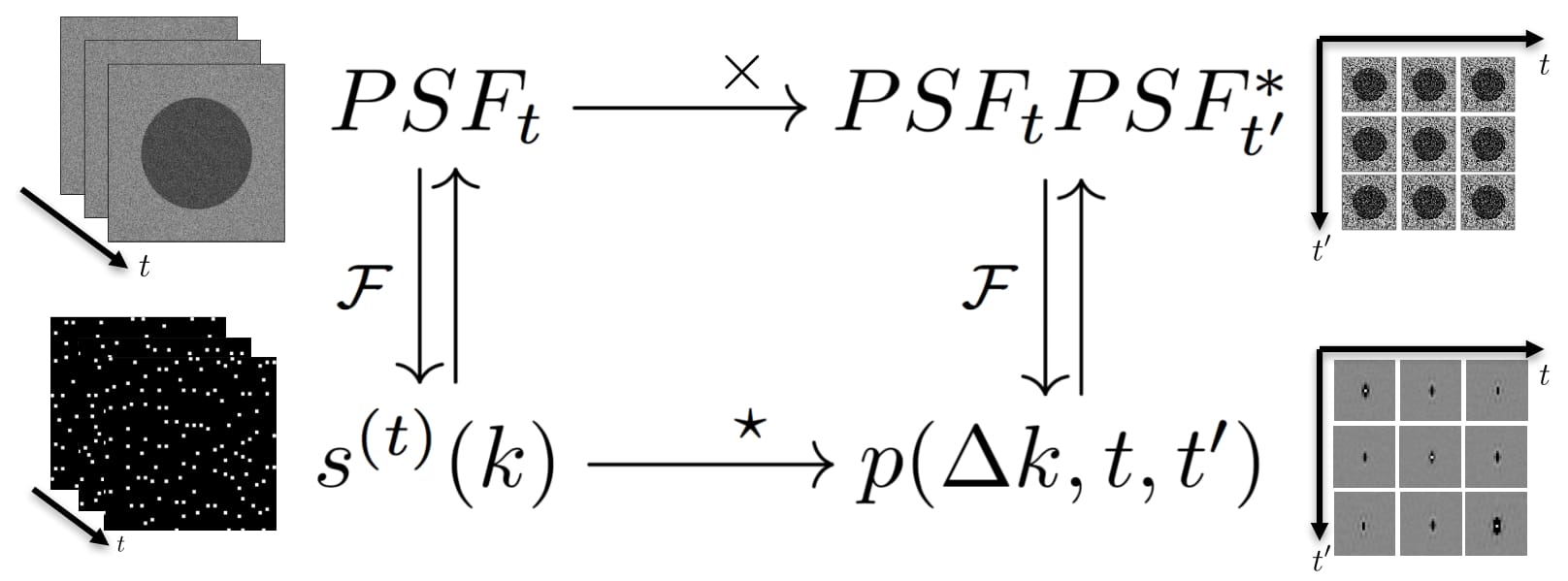}
\caption{The differential distribution is a distribution of differences in sample locations, or equivalently, a cross-correlation of sampling patterns. In image space, this correponds to a product of point-spread-functions. The differential distributions show that the sampling patterns avoid acquiring nearby samples in $k$-$t$ space, corresponding to low values of $p(\Delta k, t, t^\prime)$ where $\Delta k \approx 0$ and $t \approx t^\prime$. The central peak has been removed from point-spread-functions and their products.}
\label{fig:relationships}
\end{figure}

The differential distribution has natural properties due to the Fourier transform relationship in (\ref{eq:psf_ft2}). For example, knowing the differential distribution in a region of size $K$ near $\Delta k = 0$ summarizes the local statistics of the sampling pattern and can be used to determine products of PSFs at a resolution proportional to $1/K$. Thus, to engineer a PSF with high frequency content, one has to consider large neighborhoods of $k$-space. The dual property is that smoothing the differential distribution is equivalent to windowing the PSF. 

Examples of differential distributions are shown in Fig. \ref{fig:dd_examples}. It is natural that uniform sampling has a uniform differential distribution, since shifting the sampling pattern by the reduction factor produces the same pattern, while all other shifts produce a complementary pattern. Uniform random (white noise) sampling has a differential distribution with a constant mean where $\Delta k \neq 0$, since all pairwise differences are indepedent and uniform over the space. Poisson-disc sampling imposes a minimum distance between samples, or equivalently, that its differential distribution must satisfy
 \begin{equation}
p(\Delta k_y, \Delta k_z) = \begin{cases}
0, & \Delta k_y ^2 + \Delta k_z ^2 < \Delta k_{r,\text{max}}^2 \\
\text{arbitrary}, & \Delta k_y ^2 + \Delta k_z ^2 \geq \Delta k_{r,\text{max}}^2
\end{cases}
\end{equation}
for some minimum distance parameter $\Delta k_{r,\text{max}}$. Based on (\ref{eq:psf_ft2}), the lack of a transition band at $\Delta k_{r,\text{max}}$ explains the oscillations seen in the PSF for Poisson-disc sampling in Fig. \ref{fig:dd_examples}.

\begin{figure}[!t]
\centering
\includegraphics[width=0.49\textwidth]{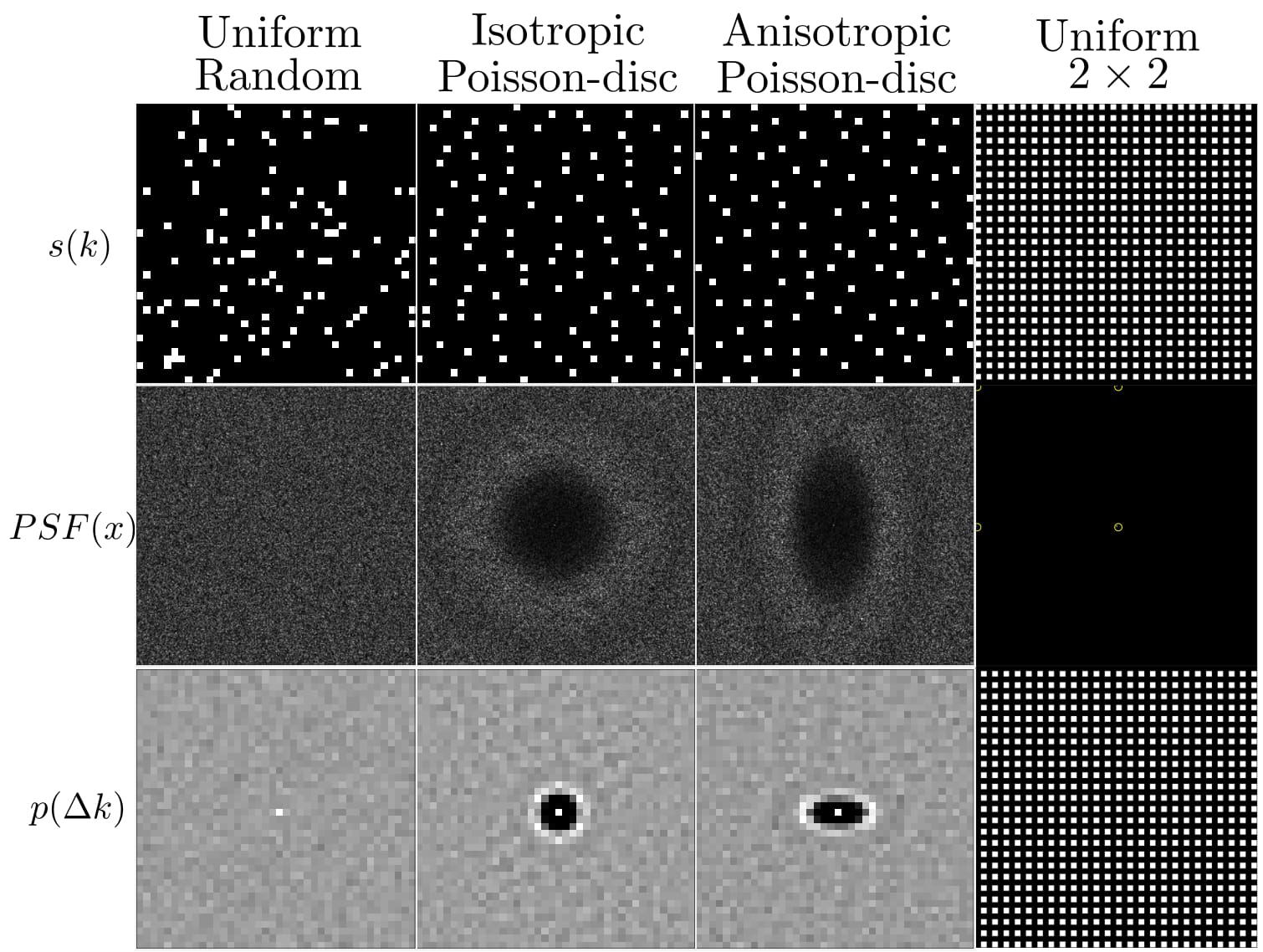}
\caption{Examples of point-spread functions and differential distributions are shown for uniform random sampling, isotropic and anisotropic Poisson-disc sampling, and separable $2\times 2$ uniform sampling. Peaks in the point-spread-functions are highlighted with yellow circles, and the central peak has been scaled for visualization.}
\label{fig:dd_examples}
\end{figure}

\subsection{Forward Model}
We consider a general linear forward model, where measurement for channel $c$ and ``frame" $t$ are generated by
\begin{equation}
y(k,t,c) = s^{(t)}(k) \mathcal{F} \{ \sum_{l} S_{t,l,c}(r) m_{l}(r) \} + \epsilon_{t,c}(k),
\label{eq:generalSENSE}
\end{equation}
where $\epsilon_{t,c}(k)$ is complex Gaussian white noise, $m_l$ is a complex-valued magnetization, $\mathcal{F}$ is a discrete Fourier transform operator, $S_{t,l,c}(r), t = 1,2,...T, c = 1, 2, \dots, C, l = 1, 2, \dots , L$ is a complex-valued ``sensitivity" function. The general case of correlated noise can be handled by a change of variables. Analogous to sensitivity maps in parallel imaging, $S_{t,l,c}$ parameterizes the subspace to which the signal in each voxel is confined, allowing one to represent any linear reconstruction in Fourier-based imaging that imposes a reconstruction subspace voxel-wise. The forward model includes dimensions along which the sampling pattern is variable (indexed by $t$) and fixed (such as coils, indexed by $c$). In voxel $r$, a $TC$-dimensional signal vector is constrained to a subspace of dimension $L$ spanned by basis functions encoded in $S_{t,l,c}(r)$ that are indexed by $l$. As described in subsequent examples, $L=T=1$ in SENSE and support-constrained imaging, and $L > 1$ and $T > 1$ accommodate more general linear reconstructions, including reconstructions for higher dimensional applications such as dynamic imaging.

(\ref{eq:generalSENSE}) can be written as a linear measurement model
\begin{equation}
    y = D\mathcal{F}Sm + \epsilon,
\end{equation}
where $D$, $\mathcal{F}$, and $S$ are operators corresponding to sampling, Fourier transformation, and sensitivity maps respectively, $\epsilon$ is a vector of complex Gaussian white noise, $m$ the magnetization, and $y$ the $k$-space data. Denoting $D\mathcal{F}S$ as $E$, a best unbiased estimate is obtained from a pseudoinverse, $(E^HE)^{-1}E^Hy$, which has a nonuniform noise term $(E^HE)^{-1} E^H \epsilon$. It is common to normalize the noise standard deviation by $\sqrt{R}$, where $R$ is the acceleration factor, and that of a fully-sampled acquisition, which is referred to as g-factor. The g-factor metric can be defined with $i$\textsuperscript{th} element
\begin{equation}
g_i = \frac{\sigma_{i,\text{accel}}}{\sigma_{i,\text{full}}\sqrt{R}},
\label{eq:gGeneral}
\end{equation}
where $\sigma_{i,\text{accel}}$ and $\sigma_{i,\text{full}}$ are the standard deviations of the $i$\textsuperscript{th} element of the fully-sampled and accelerated reconstructions respectively.

\subsection{Noise Amplification and Differential Distributions}
To select sampling patterns, an optimality criterion must be chosen that makes the information matrix $E^HE$ well-conditioned. One family of classical criteria from statistics is the spectral moments, with the $n$\textsuperscript{th} spectral moment defined as
\begin{equation}
\tr ((E^HE)^{n}) = \sum_k \lambda_k^n,
\end{equation}
where $\lambda_k$ is the $k$\textsuperscript{th} eigenvalue of $E^HE$ and $n \leq 1$. For parallel MRI, $n=-1$ and $n=-\infty$ have been investigated \cite{curtis2015optimal,xu2005optimal}.

We show formulas for the first two spectral moments for (\ref{eq:generalSENSE}). First, the spectral moment for $n=1$ is independent of the sampling locations and is given by the following expression.
\begin{thm}\label{thm:theorem1}
\begin{equation}
\sum_k \lambda_k =  \sum_{r,l,t,c} |S_{t,l,c}(r)|^2 PSF_{t}(0) 
\end{equation}
\end{thm}
\begin{proof}
A proof is provided in Appendix A. 
\end{proof}

Next, the spectral moment for $n=2$ is given by a single inner product with the differential distribution. 
\begin{thm}\label{thm:theorem2}
\begin{equation}
\sum_k \lambda_k^2 = \langle w, p \rangle,
\label{eq:obj_sos_thm2}
\end{equation}
\end{thm}
where 
\begin{equation}
w(\Delta k, t, t^\prime) = \frac{1}{N^2} \sum_{c,c^\prime} | \sum_{l} \mathcal{F} \{ S_{t^\prime,l,c^\prime}^*(r) S_{t,l,c}(r) \} |^2
\label{eq:wmtx}
\end{equation}
\begin{proof}
A proof is provided in Appendix B. 
\end{proof}

One measure of the spread in the eigenvalues is variance, which is a function of the spectral moments $\tr(E^HE)$ and $\tr((E^HE)^2)$. Minimizing the latter with the former constant minimizes the variance in the distribution of eigenvalues. In practice, $PSF_t(0)$ is often constant because the number of samples in each ``frame" is specified. Although this heuristic criterion does not have a general relationship to traditional g-factor-based criteria, its primary advantages are that it does not require matrix inversion and relates noise amplification to the sampling geometry.

Tikhonov regularization has been widely used to improve conditioning in parallel MRI and requires inverting the matrix, $E^HE + \lambda I$, where $\lambda$ is the regularization parameter. The effect of this matrix is to shift the eigenvalues by a positive constant, which does not change the variance in the distribution of eigenvalues. Thus, the criterion does not have to be modified if Tikhonov regularization is used. In conjunction with arbitrary sampling, a small amount of Tikhonov regularization is one strategy used to mitigate noise in the least squares solution that appears in later iterations of iterative linear solvers such as conjugate gradient \cite{qu2005convergence}.

$w$ encodes correlations in $k$-space that are expressed in terms of the encoding capability of the sensitivity functions. Naturally, the Fourier transform relationship implies that (anisotropic) scaling of the sensitivity functions results in the inverse scaling of $w$ and therefore the derived sampling pattern, while circulant shifts of the sensitivity functions have no effect. The bandwidth of the sensitivity functions limits the extent of the differential distribution and size of $k$-space neighborhoods that need to be considered for sampling. Section \ref{subsec:Pn} discusses the relationship between $w$ and the reproducing kernel that is associated with the model.

Theorem \ref{thm:theorem2} also provides some explanation for Poisson-disc sampling. By imposing a minimum difference between samples, Poisson-disc sampling nulls the low spatial frequencies of $w$. This confirms the intuition that Poisson-disc sampling is well-conditioned due to its uniform distribution.

Due to the use of a readout in MRI, subsets of $k$-space samples are acquired together. For Cartesian MRI, the differential distribution is constant along the readout direction. Thus, the value of $\tr((E^HE)^2)$ with a Cartesian readout is then the inner product of $w$, summed along the readout direction, and the differential distribution of the sampling pattern in the phase encoding plane. In this way, minimizing $\tr((E^HE)^2)$ for 3D or 2D Cartesian sampling is simplified to a 2D or 1D problem. 

\subsection{Efficient Optimization of Sampling Patterns}
\label{sec:optimization}  
Given a forward model parameterized by the sensitivity functions $S_{t,l,c}(\cdot)$, we formulate the sampling design as minimization of the cost-criterion
\begin{align} 
    \underset{\mathcal{S}}{\text{minimize }} J(\mathcal{S}) & = \langle w, p \rangle \label{eq:obj_sos}
\end{align}
where $\mathcal{S}$ is the set of sample locations, limited to some number, only $p$ is a function of $\mathcal{S}$, and only $w$ is a function of $S_{t,l,c}(\cdot)$ given by (\ref{eq:wmtx}). We refer to sampling patterns derived by minimizing $J(\mathcal{S})$ as $\min \tr((E^HE)^2)$ sampling.

Since computing a global optimum requires an intractable search, we instead rely on greedy heuristics. One general strategy is a sequential selection of samples, and this procedure is fast due to several properties of the differential distribution and $J(\mathcal{S})$.

Starting from an arbitrary sampling pattern and sensitivity functions, $J(\mathcal{S})$ can be computed efficiently. The differential distribution can be computed using the FFT or from pairwise differences, and $w$ can be computed with the FFT. The inner product then requires only $NT^2$ multiplications to evaluate.

To approach (\ref{eq:obj_sos}) with a greedy forward selection procedure, computing the differential distribution and cost function explicitly are not necessary. It suffices to keep track of a map of the increase in $J(\mathcal{S})$ that occurs when inserting a sample candidate $(k^\prime, t^\prime)$:
\begin{equation}
\Delta J(k^\prime,t^\prime) = w(0,t^\prime,t^\prime) + 2 \sum_{t} s^{(t)} * w(\cdot, t^\prime, t),
\label{eq:DeltaJ}
\end{equation}
where the convolution is with respect to the first argument of $w$. After inserting the sample, $\Delta J$ is updated as  
\begin{equation}
\Delta J(k,t) \rightarrow \Delta J(k,t) + 2 w(k - k^\prime, t, t^\prime).
\label{eq:DeltaJIns}
\end{equation}
Similarly, removing a sample changes the cost by $-\Delta J$. $\Delta J$ can be used to explain the sources of noise amplification in $k$-space that contribute to $J(\mathcal{S})$.

We make use of (\ref{eq:DeltaJ}) and (\ref{eq:DeltaJIns}) in an adaptation of the so-called Mitchell's best candidate sampling algorithm \cite{mitchell1991spectrally} to perform a forward sequential selection of samples, where at each step the best candidate, defined as the sample minimizing of $\Delta J$, is added to the pattern. This procedure is described in Algorithm \ref{alg:BC}.

\begin{algorithm}
\caption{Exact $\min \tr((E^HE)^2)$ Best Candidate Sampling}
\begin{algorithmic}[1]
\Procedure{ExactBestCandidate}{}
\State Initialize $\Delta J(k,t) = w(0,t,t) $ 
\State Initialize $s(k,t) \leftarrow 0$ \Comment{Sampling pattern}
\For{1..number of samples}
\LineComment{Select best candidate sample}
\State $(k^\prime,t^\prime) = \arg \min_{k,t} \Delta J(k,t)$ 
\State $s(k^\prime,t^\prime) = s(k^\prime,t^\prime) + 1$
\LineComment{Perform update from (\ref{eq:DeltaJIns}).}
\For{$k,t$}
\State $\Delta J(k,t) \leftarrow \Delta J(k,t) + 2w(k - k^\prime, t, t^\prime)$
\EndFor
\EndFor
\State \textbf{return} $s(k,t)$
\EndProcedure
\end{algorithmic}
\label{alg:BC}
\end{algorithm}

The time complexity of best candidate sampling is relatively low, but the algorithm requires on the order of a minute for single-time-point sampling with a $256 \times 256$ phase encoding matrix. Two strategies greatly improve its efficiency. One strategy, described in Section \ref{sec:approxbc}, is to use a local approximation of $w$, and by extension, the reproducing kernel, which is justifiable, for example, when the sensitivity functions are bandlimited. A second strategy is to restrict the set of sampling patterns to consider. An example is described in Section \ref{sec:periodic}.

\section{Experiments} 
\label{sec:experiments}
\subsection{Support-constrained Imaging} \label{subsec:supp}
Constraints on the object support can be represented by defining the sensitivity functions to be zero outside of the object support. Poisson-disc sampling and $\min \tr((E^HE)^2)$ sampling were compared for two support profiles generated synthetically. Although this does not account for the benefit of randomization in Poisson-disc sampling for regularization, it demonstrates the impact of an adaptive sampling strategy on conditioning compared to one common alternative. The number of samples was arbitrarily chosen to be equal to the size of the support region. The resulting noise amplification was quantified using the g-factor metric in (\ref{eq:gGeneral}), computed with a pseudo multiple replica method \cite{robson2008comprehensive} with 750 averages, a conjugate gradient solver, and a Gaussian white noise model. The following measure was used as a stopping criterion
\begin{equation}
\delta = \frac{\|x^{k+1} - x^k\|_2}{\|x^k\|_2},
\label{eq:stopping}
\end{equation}
where $x^k$ is the $k$\textsuperscript{th} iterate. The solver is stopped when $\delta$ drops below a predefined threshold. Maps of $\Delta J$ were used to identify regions of $k$-space that explain the larger value of $J(\mathcal{S})$ and the predicted sources of noise amplification. 

The first profile was a cross-shaped profile first introduced for dynamic imaging \cite{Madore}. For this profile, quincunx (checkerboard-like) sampling patterns provide optimal sampling: their g-factor is 1, and the eigenvalues of $E^HE$ are all equal. Since having all eigenvalues of $E^HE$ equal minimizes $J(\mathcal{S})$, $\min \tr((E^HE)^2)$ sampling should ideally produce quincunx sampling. For this profile, Tikhonov regularization with a very small parameter ($\lambda=10^{-4}$) was used in reconstructions. The standard deviations of Tikhnov regularized solutions were used in estimating $\sigma_{i,\text{full}}$ and $\sigma_{i,\text{accel}}$ in (\ref{eq:gGeneral}).

To demonstrate the ability to adapt to an arbitrary support region, a second support region consisting of a rotated ellipse was considered. For this profile, Tikhonov regularization with a very small parameter ($\lambda=10^{-5}$) was used.

g-Factor maps in Fig. \ref{fig:gmaps} show that $\min \tr((E^HE)^2)$ sampling patterns have lower g-factors than Poisson-disc sampling. For the cross-shaped profile, the former is identical to quincunx sampling, except in a few small regions due to the nature of local optimization, and the g-factor is very close to 1 (95\textsuperscript{th} percentile g-factor = 1.05). This can be explained by the near-perfect matching between the $w$ and $p$, which have sidebands that almost null each other. The value of $\Delta J$ is constant over $k$-space, which is consistent with the symmetry of quincunx sampling. In this case, it is possible to achieve the lower bound 
\begin{equation} 
\sum_k \lambda_k^2 \geq \frac{1}{N} (\sum_k \lambda_k)^2.
\label{eq:bounds}
\end{equation}
With equality, all eigenvalues are equal, corresponding to ideal conditioning. Naturally, it can be shown that this is possible when the support profile uniformly tiles the plane.

The elliptical region in Fig. \ref{fig:gmaps}b has a $w$ with ringing that results in a matching $p$. Maps of $\Delta J$ show regions where local noise amplification is predicted, such as in regions far from other samples seen in Poisson-disc sampling.
\begin{figure*}
\centering
\includegraphics[width=0.98\textwidth]{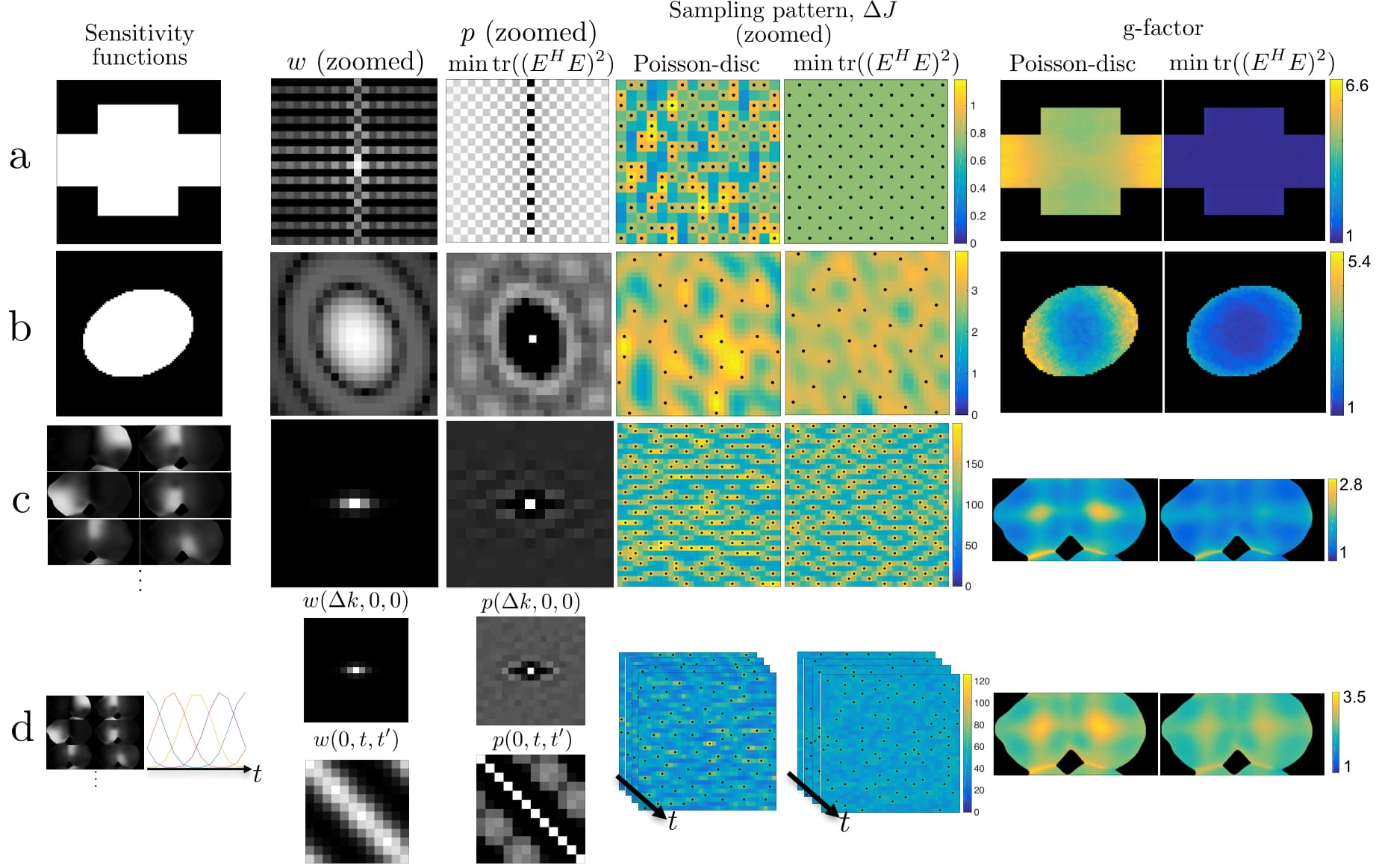} 
\caption{$\min \tr((E^HE)^2)$ and Poisson-disc sampling are compared in problems with increasing generality: support-constrained MRI (a-b), parallel MRI (c), and parallel MRI with a temporal basis of B-splines (d). Representative slices along the readout are shown. Differential distributions $p$ for $\min \tr((E^HE)^2)$ sampling patterns adapt to the weighting $w$ to minimize $\langle w, p \rangle$ and describe of geometric properties of the sampling pattern. Sampling patterns with sample locations indicated by black dots are shown with $\Delta J$, the increase in $\tr((E^HE)^2)$ resulting from sampling a new location corresponding to a prediction of local noise amplification. g-Factor maps show the noise amplification in the image domain.} 
\label{fig:gmaps}
\end{figure*}

\subsection{Parallel Imaging} \label{sec:pi}
The SENSE parallel MRI model of acquisition with an array of $C$ receiver coils can be described by a linear model composed of a spatial modulation of the magnetization by each channel sensitivity, Fourier encoding, and sampling. The channel sensitivities comprise the single set of sensitivity functions in (\ref{eq:generalSENSE}), so that $L=T=1$.

Isotropic Poisson-disc and $\min \tr((E^HE)^2)$ sampling patterns with $6$-fold acceleration were compared for SENSE reconstruction using \textit{in vivo} breast data acquired with a 16-channel breast coil (Sentinelle Medical, Toronto, Canada), a 3D spoiled gradient-echo sequence on a 3T scanner ($14^{\circ}$ flip angle, TE=2.2ms, TR=6ms, matrix size=$386 \times 192$, FOV=$22\times 17.6\times 16$ cm\textsuperscript{3}). Sensitivity maps were estimated with the ESPIRiT method \cite{uecker2014espirit} with one set of maps, scaled so that their sum-of-squares over coils is unity over the image support. g-factor maps were calculated with the pseudo multiple replica method described in Section \ref{subsec:supp} with the stopping criterion (\ref{eq:stopping}), which allowed convergence. Tikhonov regularization with a very small parameter ($\lambda=2\times 10^{-4}$) was used. $w$ was computed using the 3D sensitivity maps and summed along the readout, effectively defining the encoding operator in (\ref{eq:obj_sos}) to include the readout.

Reconstructed images and error maps for a representative slice along the readout are shown in Fig. \ref{fig:pmri_images}, which shows lower reconstruction error for $\min \tr((E^HE)^2)$ sampling (RMSE=9.6\%) than for Poisson-disc sampling (RMSE=10.6\%). This is confirmed in the g-factor maps in Fig. \ref{fig:gmaps}c. The breast coil has more coil sensitivity variation in the left-right direction, which can be seen from the extent of $w$. The differential distribution for $\min \tr((E^HE)^2)$ sampling shows that it adapts by generating higher acceleration in the left-right direction. This is also reflected in $k$-space by $\Delta J$, which is low in regions not providing optimal sampling geometry for acceleration.

\begin{figure}[!t]
\centering
\includegraphics[width=0.49\textwidth]{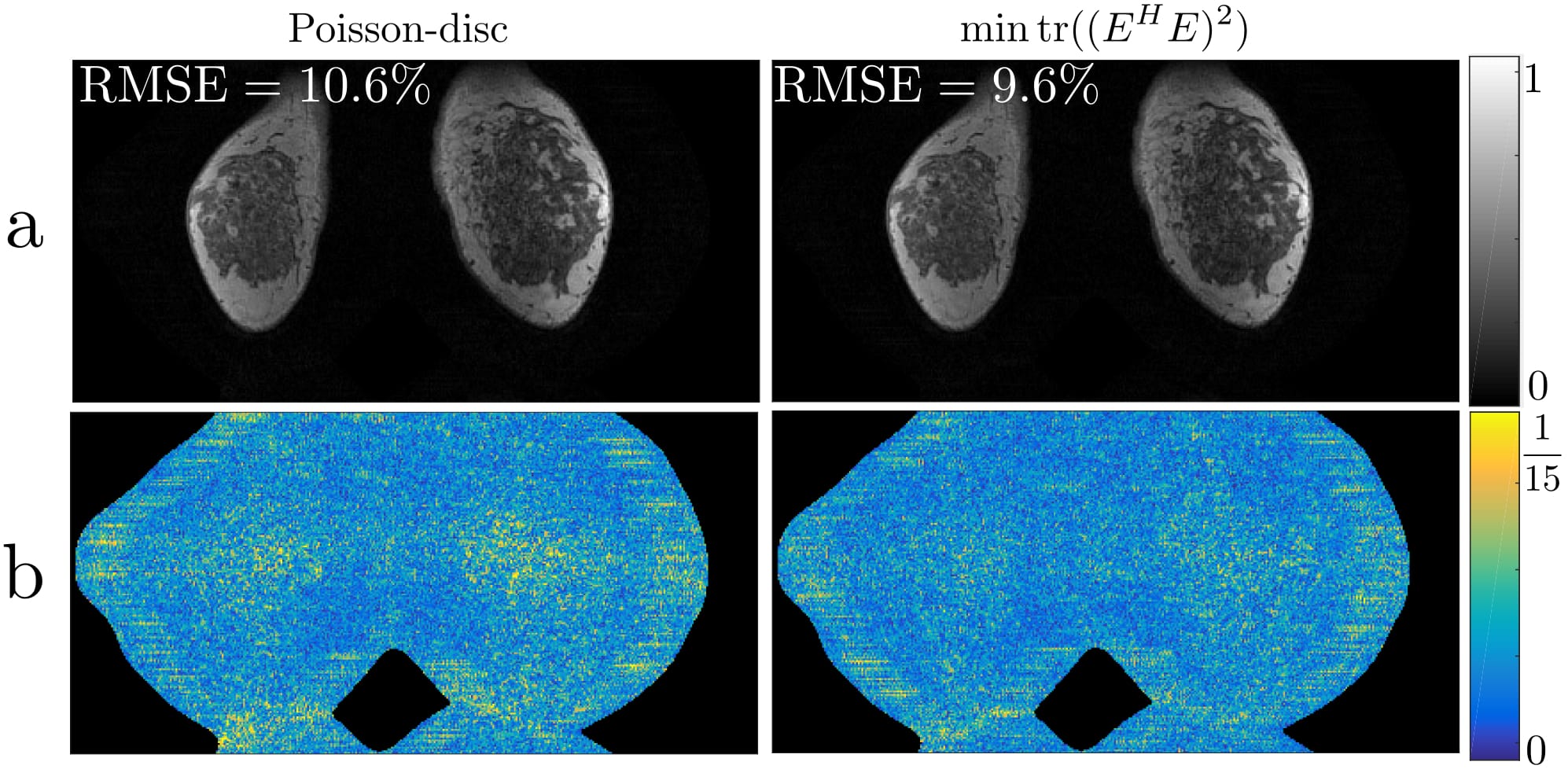}%
\caption{Reconstructed images and error maps for Poisson-disc and $\min \tr((E^HE)^2)$ sampling patterns, from experiments with parallel imaging with a breast coil. Both sampling patterns used a reduction factor of 6.}
\label{fig:pmri_images}
\end{figure}

\subsection{Temporally Constrained Dynamic Parallel MRI}
The use of $\min \tr((E^HE)^2)$ sampling for multidimensional imaging was demonstrated for multi-time-point (dynamic) parallel MRI. For dynamic imaging to be well-posed, it is necessary to assume a certain limited temporal resolution. One generic way to incorporate this assumption to to represent the dynamics in each voxel by a predefined basis of B-splines, allowing coefficients to be recovered by linear reconstruction \cite{nichols2002spatiotemporal,jerosch2002myocardial,filipovic2011motion,le2016spline}. 

Assuming a discretization of the time axis into $T$ frames indexed by $t$, and $L<T$ basis functions $\{ \beta_l(t) \}_{l=1}^L$, (\ref{eq:generalSENSE}) can be used to represent such a contrast change model with 
\begin{equation}
S_{t,l,c}(r) = C(r,c) \beta_l(t),
\label{eq:temporalbasis2}
\end{equation}
where $C(r,c)$ are the coil sensitivity maps for channel $c$. The $l$\textsuperscript{th} image of coefficients in this basis is $m_l(r)$. This can also be viewed as a subspace model \cite{liang2007spatiotemporal} with a predefined temporal basis. The sampling problem is to select a $k$-$t$ sampling pattern that balances spatiotemporal correlations so that inversion of the model in (\ref{eq:temporalbasis2}) is well-conditioned. For (\ref{eq:temporalbasis2}), the optimal strategy suggested by $\tr((E^HE)^2)$ is given by
\begin{equation}
w(\Delta k, t, t^\prime) = \lbrack |\sum_{l} \beta_l(t) \beta_l^*(t^\prime)|^2 \rbrack \lbrack \frac{1}{N^2} \sum_{c,c^\prime} | \mathcal{F} \{ C_{c^\prime}^*(r) C_{c}(r) \} |^2 \rbrack,
\end{equation}
which is a product of two functions in brackets, the first of which encodes correlations in time and the second of which encodes correlations in $k$-space. These functions are shown in Fig. \ref{fig:gmaps}d.

The breast data and coil sensitivities used in experiments with parallel imaging was used in a hypothetical dynamic MRI acquisition with a linear reconstruction of the coefficients in the spline basis. The time axis was discretized into $T=12$ frames, and a basis of $L=4$ third-order B-splines with periodic boundary conditions and maximum value of unity was used to represent all signal intensity-time curves. Poisson-disc sampling and $\min \tr((E^HE)^2)$ sampling with $R=12$ were compared, and an additional constraint was used to make the number of samples in each frame constant. Poisson-disc sampling patterns were generated independently for each frame. A temporally averaged g-factor was calculated using the pseudo multiple replica method described in Section \ref{subsec:supp} followed by a mean of the $L$ coefficient frames. Tikhonov regularization was used with a very small parameter ($\lambda=2\times 10^{-6}$).

The plot of $w$ sliced in two planes shows a representation of correlations in $k$-$t$-space. g-Factor is lower in $\min \tr((E^HE)^2)$ sampling due to the ability to balance these correlations with the improved match between $w$ and $p$ (Fig. \ref{fig:gmaps}d). This also corresponds with a more uniform $\Delta J$.

\subsection{Empirical Comparison of Optimality Criteria}
To compare the proposed and conventional optimality criteria for selecting arbitrary sampling patterns, sampling patterns minimizing $\tr((E^HE)^2)$ and mean squared error (MSE, given by $\tr((E^HE)^{-1})$), following the approach of \cite{xu2005optimal} and a uniform sampling pattern were compared for accelerated 2D fast spin echo \textit{in vivo} human knee data acquired from a GE 3T scanner (TE=7.3ms, TR=1s, echo train length=8, slice thickness=3mm, resolution=$1.25 \times 1.25$ mm\textsuperscript{2}, 8 channel knee coil). A reduction factor of 2 was used for all sampling patterns. Optimal sampling patterns for the MSE criterion were generated using the method described in \cite{xu2005optimal}, which reported a computation time of 5 minutes for the same matrix size. $w$ was computed from 2D sensitivity maps and summed along the readout. 
g-Factor maps were calculated with the pseudo multiple replica method described in Section \ref{subsec:supp} with the stopping criterion (\ref{eq:stopping}), which allowed convergence. Tikhonov regularization was used with a very small parameter ($\lambda=10^{-4}$).

For the proposed criterion, sampling patterns were generated in 367 ms, which included computation of $w$ in 304 ms and Algorithm \ref{alg:BC} in 63 ms. Fig. \ref{fig:oscar} shows sampling patterns, g-factor maps, and reconstructed images for the three sampling patterns. Commonly reported metrics for sampling patterns and the time to generate them are shown in Table \ref{tab:oscar}. MSE was calculated using images reconstructed from the original and subsampled data. The two optimality criteria yield similar results, indicating that the criteria are empirically equivalent. The small differences between the two patterns in all metrics and lower MSE for the proposed criterion are expected due to the nature of greedy methods. However, the time to generate the sampling pattern using the proposed criterion is much shorter due to the use of the inverse matrix in computing MSE.

Both criteria show an improvement over uniform sampling, which results in a small region of very high g-factor at the center of the image. This region shows marked noise artifacts in the zoom-in images. Although uniform sampling is optimal in many cases, uniform sampling can be suboptimal and similar results have previously been observed for the MSE criterion \cite{xu2005optimal}.

\begin{figure*}[!t]
\centering
\includegraphics[width=0.98\textwidth]{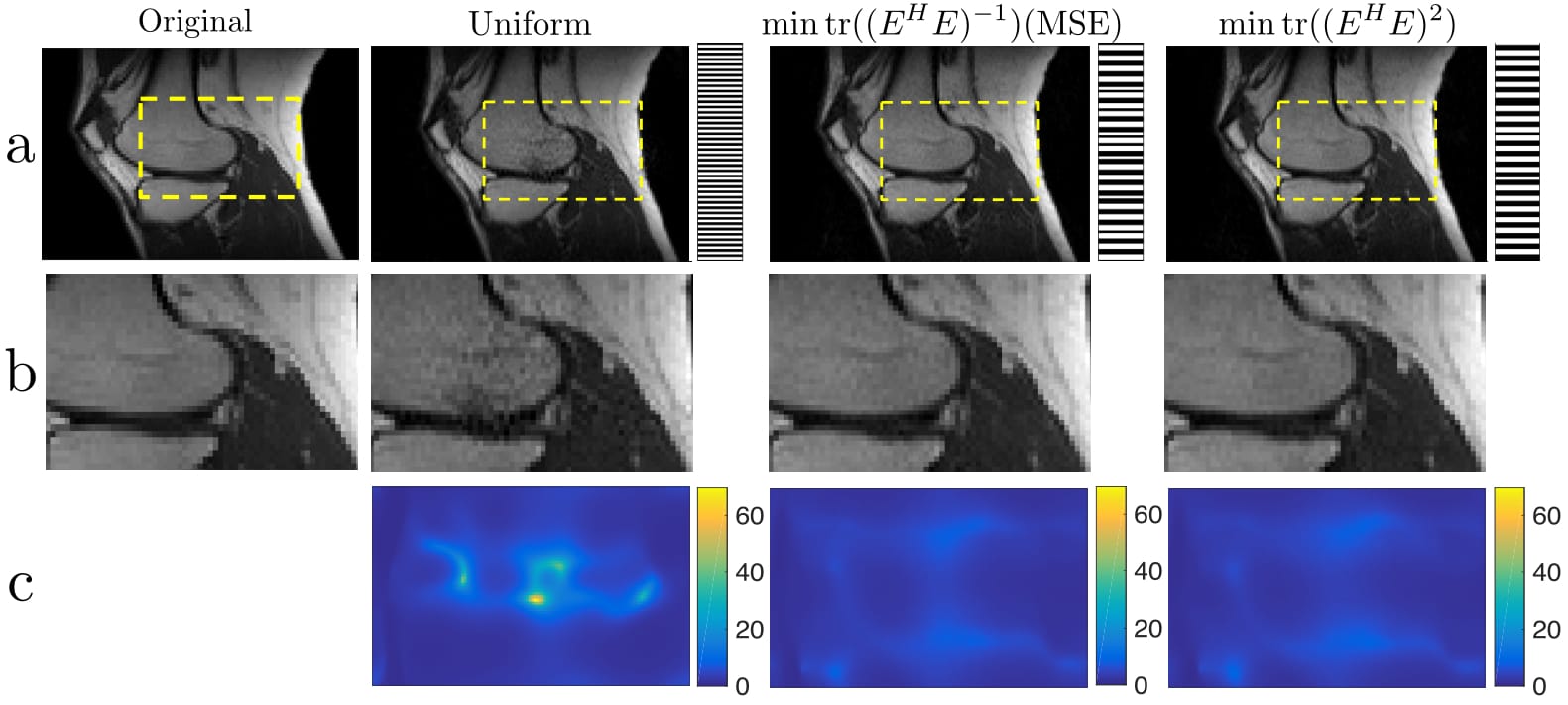}%
\caption{Two optimality criteria, mean squared error and the proposed, $\tr((E^HE)^2)$, are compared in parallel MRI reconstruction of 2D \textit{in vivo} human knee data. (a) Reconstructed images and sampling patterns appear similar and both show an improvement over uniform sampling. (b) Zoom-in images show a reduction in artifacts at the center of the image with either criteria, which is confirmed by (c) g-factor maps shown for the full image. Metrics and execution times are given in Table \ref{tab:oscar}.}
\label{fig:oscar}
\end{figure*}

\begin{table}
   \centering
    \caption{Metrics for sampling patterns compared in Fig. \ref{fig:oscar} and time to generate.}
    \label{tab:oscar}
    \begin{tabular}{c|c|c|c}
        \hline
                          & Uniform & min tr$((E^HE)^{-1})$ & min tr$((E^HE)^2)$ \\ 
                          &  & (MSE) & (Proposed)  \\ \hline \hline
        MSE               & 1       & 0.41                        & 0.38              \\ 
        tr$((E^HE)^2)$    & 1       & 0.978                       & 0.976              \\ 
        Max g-factor      & 69.6    & 10.1                        & 9.3                \\ 
        Median g-factor   & 3.14    & 3.05                        & 2.97               \\ 
        Time to generate  & 0       & 5 minutes                   & 0.37 seconds       \\ \hline \hline
    \end{tabular}
\end{table}

\subsection{Reconstruction Errors in $k$-Space} \label{subsec:Pn}
The analysis of \cite{athalye2015parallel} provides a description of parallel MRI as approximation in a reproducing kernel Hilbert space and a characterization of reconstruction errors in $k$-space. From this analysis, it follows that the space of $k$-space signals in (\ref{eq:generalSENSE}) is a reproducing kernel Hilbert space, for which the discrete matrix-valued reproducing kernel is
\begin{equation}
K_{ct,c^\prime t^\prime}(k,k^\prime) = \frac{1}{N} \sum_{r} (\sum_{l=1}^L S_{t,l,c}(r) S_{t^\prime,l,c^\prime}^*(r))e^{-\frac{2 \pi i (k - k^\prime) \cdot r}{N}}. 
\label{eq:repKernel}
\end{equation}
Like $K$, $w$ also represents correlations in $k$-space and is related to $K$ by
\begin{equation}
w(\Delta k, t, t^\prime) = \sum_{c,c^\prime} | K_{ct,c^\prime t^\prime}(\Delta k, 0) |^2. 
\label{eq:wKRelationship}
\end{equation}
As described in \cite{athalye2015parallel}, local approximation errors can be predicted by a \textit{power function}, which is computed from $K$ and the sample locations by solving a large linear system. For notational simplicity, we drop $t$ and $l$ indices in the following discussion. The power function is given by
\begin{equation}
P_{c}^2(k) =  K_{c,c}(0,0) + \sum_{n} \sum_{c^\prime=1}^C K_{c,c^\prime}(k,k_n) u_{c}^{*c^\prime n}(k),
\label{eq:Pfunc}
\end{equation}
where the first sum is over the sample locations $k_n$, $u_{c}^{c^\prime n}(k)$ are the cardinal functions (interpolation weights), defined in \cite{athalye2015parallel}. A coil-combined power function can be computed by a sum of $P_{c}^2(k)$ over coils.

One potential application of the power function is to guide the design of sampling patterns. However, the primary challenge is that the cardinal functions require solving large systems of linear equations, making computation of the power function very expensive. It is important to note that for the purpose of determining optimal sample locations, it is not necessary to compute quantitative error bounds. As a computationally inexpensive, albeit semiquantitative, metric, $\Delta J$ may provide the same information as a coil-combined power function. 

$\Delta J$ and the power function were compared for the experiments with parallel imaging described in Section \ref{sec:pi}. To reduce the computation time for the power function, the original breast data was cropped to a lower resolution (matrix size=$96\times 48$). The coil-combined power function was computed by summing over coils. $\Delta J$ and the power function were compared for uniform random and $\min \tr ((E^HE)^2)$ sampling patterns.

Fig. \ref{fig:Pn}a and b show that $\Delta J$ and the power function make similar predictions, while their relationship is nonlinear. This is confirmed in Fig. \ref{fig:Pn}c, which shows that the two metrics are highly correlated (Spearman's rank correlation = 0.994 for both sampling patterns). It is tempting to explain this by analyzing second term in (\ref{eq:Pfunc}), which may appear similar to the convolution in (\ref{eq:DeltaJ}). However, they are only proportional when the power function is summed over coils and $u_{c}^{c^\prime n}(k) = K_{c,c^\prime}(k,k_n)$, which corresponds to replacing the interpolation weights with a single convolution with the kernel. Although it is clear that a general relationship between the two metrics does not exist, their high correlation suggests that $\Delta J$ is a good surrogate for the power function for some purposes. This relationship makes the proposed method approximately equivalent to the sampling procedure suggested in \cite{athalye2015parallel} based on selecting samples at minima of the power function.

\begin{figure*}[!t]
\centering 
\includegraphics[width=0.98\textwidth]{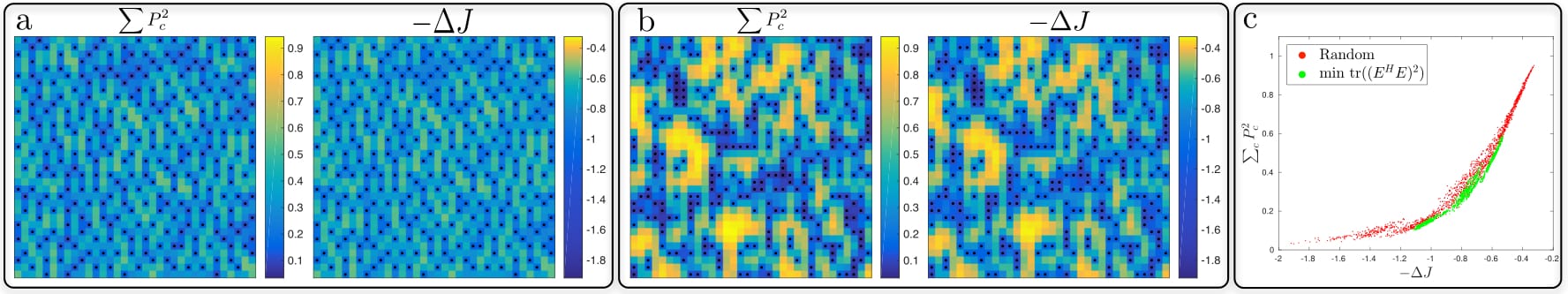}%
\caption{The combined power function $\sum_c P_c^2$ and $\Delta J$ for $\min \tr ((E^HE)^2)$ (a) and random (b) sampling patterns were compared for parallel imaging reconstruction with \textit{in vivo} human breast data. Sample locations are shown as block dots with each metric overlaid. Plotting the combined power function against $\Delta J$ from all $k$-space locations shows a nonlinear, monotic trend indicating that they provide similar information in this example (Spearman's rank correlation = 0.994 for both patterns). Note that the color scale in (a) and (b) has been adjusted by fitting a monotonic function to the data in (c).}
\label{fig:Pn} 
\end{figure*}

\subsection{Approximate Best Candidate Sampling}  \label{sec:approxbc} 
One strategy for speeding up Algorithm \ref{alg:BC} is to modify $w$ to have limited support, which is possible when the sensitivity functions are bandlimited. This assumption has been successfully used in procedures for automatic selection of sampling patterns that are based on computing interpolation weights \cite{samsonov2008optimality, samsonov2009automatic}. Similarly, $w$ can be replaced with a sparse surrogate $\hat{w}$ by thresholding. The number of inner-loop iterations is then reduced from $NT$ to $\text{supp}(\hat{w})$, where $\text{supp}(\hat{w})$ is the number of nonzeros in $\hat{w}$.

With this modification, a naive implementation still has time complexity $O(NTS)$, where $S$ is the number of samples, to compute the minimum of $\Delta J$. However, this step can be carried out efficiently using a priority queue. The inner-loop is performed by modifying only the region of $k$-$t$ space corresponding to the support of $w$. An implementation is summarized in Algorithm \ref{alg:approxBC}. A total of $2 \text{supp}(\hat{w})$ elements must be incremented, and incrementing each requires $O(\log(NT))$ complexity, resulting in a total computional complexity of $O(\text{supp}(\hat{w}) S \log(NT))$. If $\text{supp}(\hat{w})$ is small enough, this algorithm is faster than Algorithm \ref{alg:BC}. 

For the experiment in Section \ref{sec:pi}, the run time of the two algorithms is compared in Fig. \ref{fig:approxbc}. When the support of $\hat{w}$ is near 20,000 ($0.3N$), the run times for Algorithms \ref{alg:BC} and \ref{alg:approxBC} are equal. In general, this crossing point depends on $T$ and $N$, and for $T$ large, the support of $\hat{w}$ must be smaller. When the support of $\hat{w}$ is only 16 ($0.0002N$), this results in nearly optimal g-factor and objective value, which shows that best candidate sampling can be dramatically simplified without sacrificing performance. At this level of approximation, the run time is near one second, which is sufficient for interactive applications. 

\begin{algorithm}
\caption{Approximate $\min \tr((E^HE)^2)$ Best Candidate Sampling}
\begin{algorithmic}[1]
\Procedure{ApproximateBestCandidate}{}
\State Set $\hat{w}$ to a thresholded version of $w$
\State Initialize $\Delta J(k,t) = \hat{w}(0,t,t) $
\State Initialize $s(k,t) \leftarrow 0$
\State Construct priority queue with all $k$-$t$-space locations as keys and corresponding $\Delta J$ as values \Comment{$O(NT)$}
\For{1..number of samples}
\State $(k^\prime,t^\prime) = \arg \min_{k,t} \Delta J(k,t)$ \Comment{$O(\log{(NT)})$}
\State $s(k^\prime,t^\prime) = s(k^\prime,t^\prime) + 1$
\LineComment{First term in (\ref{eq:DeltaJIns}).}
\For{$(k,t) \in \{ (k,t) : \hat{w}(k - k^\prime, t, t^\prime) \neq 0 \}$} 
\State $\Delta J(k,t) \leftarrow \Delta J(k,t) +  \hat{w}(k - k^\prime, t, t^\prime)$ 
\EndFor
\LineComment{Second term in (\ref{eq:DeltaJIns}).}
\For{$(k,t) \in \{ (k,t) : \hat{w}(k^\prime - k, t^\prime, t) \neq 0 \}$}
\State $\Delta J(k,t) \leftarrow \Delta J(k,t) +  \hat{w}(k^\prime - k, t^\prime, t)$ 
\EndFor
\EndFor
\State \textbf{return} $s(k,t)$
\EndProcedure
\end{algorithmic}
\label{alg:approxBC}
\end{algorithm}

\begin{figure*}[!t]
\centering
\includegraphics[width=0.98\textwidth]{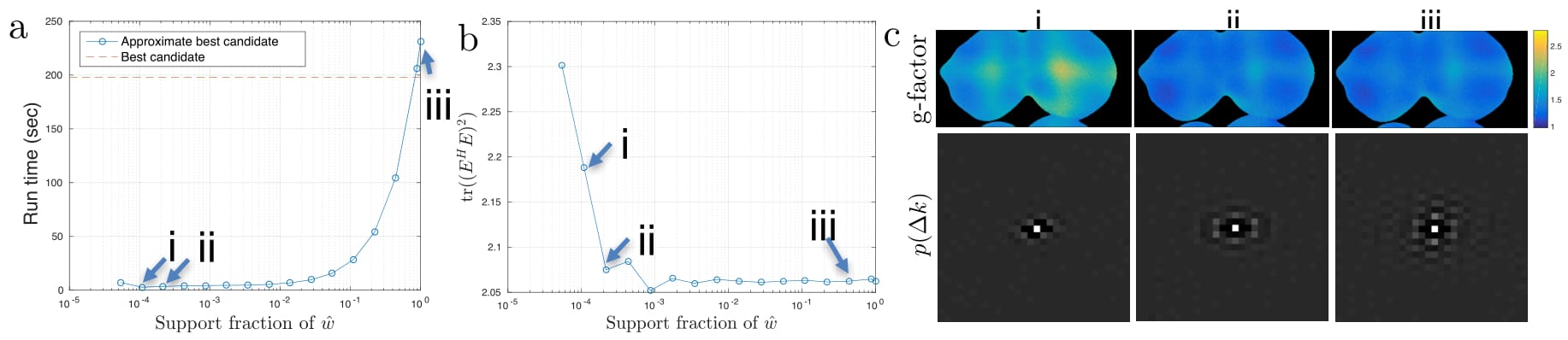}
\caption{For efficiency, best candidate sampling can be approximated by replacing $w$ with a thresholded version $\hat{w}$. Reducing the support of $\hat{w}$ results in (a) shorter run times with (b) nearly-optimal values of $\tr((E^HE)^2)$ until $\hat{w}$ is too small to accurately represent correlations. g-Factors and differential distributions are shown in (c) for three points on the tradeoff, where point ii achieves similar g-factor to iii but has much shorter run time ($\approx 1$ second). For point i, $w$ is not well-approximated by $\hat{w}$, resulting in suboptimal sampling. Increasing the support of $\hat{w}$, the differential distribution shows variation over a larger region.}
\label{fig:approxbc}
\end{figure*}

\subsection{Periodic Sampling} \label{sec:periodic} 
To enable fast selection of sampling patterns, one strategy is to consider only a subset, such as periodic sampling patterns. If a sampling pattern can be expressed as a convolution of two sampling patterns, its differential distribution is the convolution of the differential distributions for the two sampling patterns. A periodic sampling pattern can be represented as a convolution of a lattice of $\delta$ functions with period $n_0$, which we denote $\Sh$ and a cell $s_0^{(t)}$ of size $n_0$: 
\begin{equation}
s^{(t)} = \Sh * s_0^{(t)}
\end{equation}

Let $p_0(\Delta k, t, t^\prime)$ be the differential distribution of $s_0^{(t)}$. Then 
\begin{equation}
p(\Delta k, t, t^\prime) = \frac{N}{n_0} \Sh * p_0(\cdot, t, t^\prime).
\label{eq:p_caip}
\end{equation}
Since $p$ is periodic, only the period of $p$, $p_0$, needs to be computed. Since a small number of periodic sampling patterns need to be considered, and often significantly fewer than the number of unique ones, it is feasible to exhaustively evaluate periodic sampling patterns.

Differential distributions can be efficiently computed using (\ref{eq:p_caip}). 2D-CAIPIRINHA \cite{Breuer2006} sampling patterns with an acceleration factor of $R$ are periodic with a period of size $R\times R$. The acceleration factors in $y$ and $z$ directions, $R_y$ and $R_z$ with $R_yR_z=R$, must be chosen in addition to one shift parameter.

All 12 2D-CAIPIRINHA sampling patterns with $R=6$ were evaluated using g-factor maps and $\Delta J$, with the exception of 4 that resulted in extremely high g-factors. g-Factor maps were computed using the analytical expression \cite{pruessmann1999sense}. Spearman's rank correlation was used to show agreement between $\tr((E^HE)^2)$ and several g-factor metrics, including the maximum g-factor which was used previously \cite{deshpande2012optimized,Weavers2014}. Fig. \ref{fig:caip} shows that the pattern with $3 \times 2$ acceleration in $k_y$ and $k_z$ has the most uniform $\Delta J$ and lowest value of $\tr((E^HE)^2)$. The rank correlations between $\tr((E^HE)^2)$ and the mean, maximum, and root-mean-square g-factor were $0.93$, $0.95$, and $0.90$ respectively, indicating strong agreement between $\tr((E^HE)^2)$ and g-factor-based criteria.

\begin{figure*}[!t]
\centering
\includegraphics[width=0.99\textwidth]{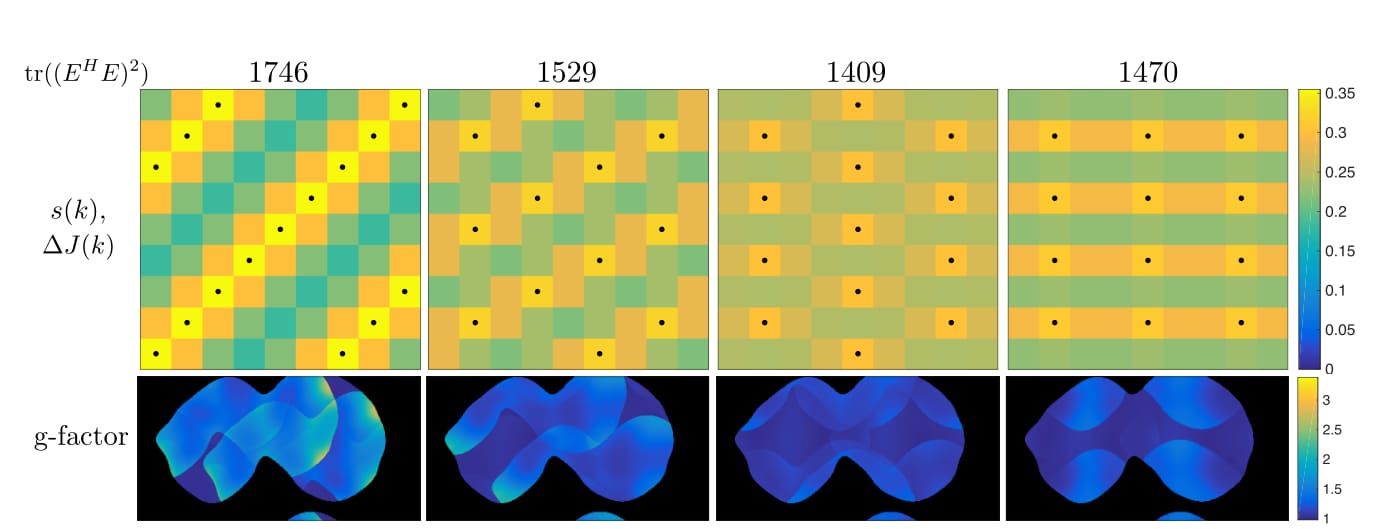}
\caption{Periodicity of the sampling pattern is preserved in the differential distribution, allowing efficient evaluation of periodic sampling patterns. 2D-CAIPIRINHA sampling patterns with $R=6$ are compared with samples indicated by black dots and, $\Delta J$ showing local contributions to $\tr((E^HE)^2)$, along with corresponding g-factor maps are shown. Strong agreement was observed between $\tr((E^HE)^2)$ and several g-factor metrics (Spearman's rank correlation = 0.90, 0.93, and 0.95 for mean, max, and root-mean-square g-factors).}
\label{fig:caip}
\end{figure*}

\section{Summary and Conclusion} 
A novel concept of a differential distribution was shown to be related to conditioning in linear reconstruction via moment-based spectral analysis of the information matrix, leading to a new criterion for sampling. This formalism describes a natural link between point-spread-functions, geometric features of sample distributions, and conditioning for linear reconstruction. It also has unique computational advantages that enable efficient, automatic design of arbitrary and adaptive multidimensional Cartesian sampling patterns. These advantages are threefold: 1) the differential distribution is defined close to the domain where sampling patterns are designed, allowing basic operations to be carried out efficiently, 2) matrix inversion is avoided, and 3) local dependencies are compactly represented in the differential distribution. Greedy algorithms exploiting these advantages are presented and shown to achieve very fast execution times, often sufficient for real-time interactive applications.

The present approach has several limitations. First, only linear reconstructions are considered. Incorporating arbitrary regularization requires considering both conditioning and additional criteria that may lead to variable density or incoherent aliasing. Although many reconstructions in MRI are not linear, the framework may provide some guidance for inverse problems that are bilinear, linear in a subset of variables, or can be associated with a linear inverse problem. Second, the variance in the distribution of eigenvalues has no general relationship to the g-factor or power function, which are comprehensive and quantitative measures of noise amplification but are computationally expensive. Numerical experiments with \textit{in vivo} human data illustrate strong agreement between the criterion and these metrics. This criterion is not a replacement for these metrics but rather, a complementary metric that has utility for sampling. Empirically, the criterion does not select pathological eigenvalue distributions that have a small variance but are still ill-conditioned. Third, while the theory readily extends to non-Cartesian imaging, only Cartesian imaging is considered in the present work. Similar Algorithms for non-Cartesian sampling would be trajectory-dependent. We note that these greedy approaches do not guarantee a globally optimal solution. Finally, arbitrary sampling patterns often require some regularization that is not required in reconstruction from uniform subsampling.

The implications of the proposed criterion have been investigated on multiple fronts. Numerical experiments in support-constrained, parallel, and dynamic MRI have demonstrated the adaptive nature of the approach and agreement between the criterion and g-factor. Further experimental validation in 2D and 3D Cartesian MRI with arbitrary and periodic sampling patterns shows agreement with existing metrics, including MSE, and the $k$-space-based power function metric. The formalism has been related to an approximation theory formalism \cite{athalye2015parallel}, which offers similar descriptions of correlations in $k$-space and predictions about local reconstruction errors. The approach and results in this paper indicate the potential for computationally feasible, arbitrary, and adaptive sampling techniques for linear, multidimensional MRI reconstruction.

\section*{Appendix A: Proof of Theorem \ref{thm:theorem1}}
\label{app:theorem1}
\begin{proof}[\unskip\nopunct]
Let $S_{tc,l}$ be the diagonal operator corresponding to point-wise multiplication by $S_{t,c,l}(r)$ and $D_t$ be the diagonal sampling operator for frame $t$. 

Block $ij$ of the matrix $E^HE$ is 
\begin{equation}
(E^HE)_{ij} = \sum_{t=1}^T \sum_{c=1}^C S_{tc,i}^H F^H D_t F S_{tc,j} 
\end{equation}
Then
\begin{align}
\tr(E^HE) &= \sum_{l=1}^L \tr( (E^HE)_{ll} ) \\
 &= \sum_{l=1}^L \sum_{t=1}^T \sum_{c=1}^C \tr( S_{tc,l}^H F^H D_t F S_{tc,l} ) \\
 &= \sum_{l=1}^L \sum_{t=1}^T \sum_{c=1}^C \sum_r |S_{t,l,c}(r)|^2 e_r^H F^H D_t F e_r \\
 &= \sum_{l=1}^L \sum_{t=1}^T \sum_{c=1}^C \sum_r |S_{t,l,c}(r)|^2 PSF_t(0)
\end{align}
\end{proof}

\section*{Appendix B: Proof of Theorem \ref{thm:theorem2}}
\label{app:theorem2}
\begin{proof}[\unskip\nopunct]
Let $S_{tc,l}$ be the diagonal operator defined in Appendix A.
Then 
\begin{align}
&\tr( (E^HE)^2 ) \nonumber \\
=& \sum_{l=1}^L \tr( \sum_{l^\prime=1}^{L} (E^HE)_{l  l^\prime}(E^HE)_{l^\prime l} ) \label{eq:thm2_step1} \\
=& \sum_{l,l^\prime,t,t^\prime,c,c^\prime} \tr( S_{tc,l}^H F^H D_t F S_{tc,l^\prime} S_{t^\prime c^\prime,l^\prime}^H F^H D_{t^\prime} F \nonumber \\
 & S_{t^\prime c^\prime,l} ) \label{eq:thm2_step3} \\ 
=& \sum_{l,l^\prime,t,t^\prime,c,c^\prime,r,r^\prime} PSF_t(r-r^\prime) PSF_{t^\prime}^*(r-r^\prime) \nonumber \\
& S_{t,c,l^\prime}(r^\prime) S_{t^\prime, c^\prime,l^\prime}^*(r^\prime) S_{t,c,l}^*(r) S_{t^\prime,c^\prime,l}(r) \label{eq:thm2_step4} \\
=& \sum_{l,l^\prime,t,t^\prime,c,c^\prime} \langle (S_{t,l,c} \cdot S_{t^\prime,l, c^\prime}^*)(r), ((PSF_t \cdot PSF_{t^\prime}^*) \nonumber \\
 & * (S_{t,l^\prime,c} \cdot S_{t^\prime,l^\prime,c^\prime}^*))(r) \rangle \label{eq:thm2_step5} \\
=& \sum_{l,l^\prime,t,t^\prime,c,c^\prime} \langle \mathcal{F}\{ S_{t,l,c} S_{t^\prime, l, c^\prime}^* \}, p(\Delta k,t,t^\prime) \cdot \nonumber \\
 & \mathcal{F}\{ S_{t,l^\prime,c} \cdot S_{t^\prime,l^\prime,c^\prime}^* \} \rangle \label{eq:thm2_step6} \\
=& \sum_{l,l^\prime,t,t^\prime,c,c^\prime} \langle | \sum_{l=1}^L \mathcal{F}\{ S_{t,l,c}\cdot S_{t^\prime, l, c^\prime}^* \} |^2, p(\Delta k, t,t^\prime) \rangle \label{eq:thm2_step7} \\
=& \langle w, p \rangle \label{eq:thm2_step8}
\end{align}
Parseval's theorem and the convolution theorem are applied in (\ref{eq:thm2_step6}). The inner product in (\ref{eq:thm2_step5}), (\ref{eq:thm2_step6}-\ref{eq:thm2_step7}), and (\ref{eq:thm2_step8}) are sums over $r$, $\Delta k$, and $(\Delta k, t, t^\prime)$ respectively.
\end{proof}

\section*{Acknowledgment}
The authors would like to thank Dr. V. Taviani for help in experiments with in vivo knee data.



%

\bibliographystyle{IEEEtran}
\bibliography{IEEEabrv,bibliography}

\end{document}